\documentclass[reqno]{amsart}
\usepackage{comment,amssymb,latexsym,upref,enumerate,fouridx}
\usepackage{finalT}
\numberwithin{equation}{section}

\input Tdef.def

\newcommand{\Au}{\mathbf{A}}
\newcommand{\Acu}{\boldsymbol{\mathcal{A}}}
\newcommand{\Pbbu}{\underline{\mathbb{P}}{}}
\newcommand{\ufu}{\boldsymbol{\mathfrak{u}}{}}
\newcommand{\vfu}{\boldsymbol{\mathfrak{v}}{}}

\begin{document}

\title[Future stability of the FLRW fluid solutions]{Future stability of the FLRW fluid solutions in the presence of a positive cosmological constant}

\author[T.A. Oliynyk]{Todd A. Oliynyk}
\address{School of Mathematical Sciences\\
Monash University, VIC 3800\\
Australia}
\email{todd.oliynyk@monash.edu}

\begin{abstract}
\noindent  We introduce a new method for establishing the future non-linear stability of perturbations of FLRW solutions to the Einstein-Euler
equations with a positive cosmological constant and a linear equation of state of the form $p = K \rho$. The method is
based on a conformal transformation of the Einstein-Euler equations that compactifies the time domain and can handle
the equation of state parameter values $0<K\leq 1/3$ in a uniform manner. It also determines the
asymptotic behavior of the perturbed solutions in the far future.
\end{abstract}

\maketitle


\sect{intro}{Introduction}
The Einstein-Euler equations for an isentropic perfect fluid with a linear equation of
state and a positive cosmological constant
are given by
\lalin{EE}{
&\Gt^{ij} = \Tt^{ij}-\Lambda \gt^{ij} \qquad (\Lambda > 0),\label{EE.1}\\
&\nablat_i \Tt^{ij} = 0, \label{EE.2}
}
where $\Gt^{ij}$ is the Einstein tensor of the metric
\leqn{gtdef}{
\gt = \gt_{ij} d x^i d x^j,
}
$\Lambda$ is the cosmological constant, and
\eqn{Ttdef}{
\Tt^{ij} = (\rho + p)\vt^i \vt^j + p \gt^{ij} \qquad (\gt_{ij}\vt^i\vt^j = -1)
}
is the perfect fluid stress energy tensors with the pressure determined by the equation of state
\leqn{eosdef}{
p = K \rho \qquad (K\geq 0).
}
On spacetimes of the form $[T_0,\infty)\times \Tbb^3$, the Friedmann-Lema\^{i}tre-Robertson-Walker (FLRW) solutions
to \eqref{EE.1}-\eqref{EE.2} represent a homogenous, fluid filled universe that is undergoing accelerated expansion.
In \cite{RodnianskiSpeck:2013}, the future non-linear stability of these solutions under the condition $0<K<1/3$ and the assumption of zero fluid vorticity was
established by Rodnianski and Speck using techniques developed by Ringstr\"{o}m in \cite{Ringstrom:2008}.
Subsequently, it has been shown \cite{HadzicSpeck:2015,LubbeKroon:2013,Speck:2012}
that this future non-linear stability result remains true for fluids with non-zero vorticity and
also for the equation of state parameter values $K=0$ and $K=1/3$, which correspond to dust and pure radiation, respectively.
It is worth noting that the stability result for $K=1/3$ established in \cite{LubbeKroon:2013} relies on Friedrich's conformal method \cite{Friedrich:1986,Friedrich:1991},
which is completely different from the techniques used in \cite{HadzicSpeck:2015,RodnianskiSpeck:2013,Speck:2012} for the parameter values $0\leq K<1/3$.

In this article, we introduce a new method for establishing the future non-linear stability of FLRW solutions
to \eqref{EE.1}-\eqref{EE.2}. The advantage of our method is threefold: (i) it leads to a relatively compact
stability proof, (ii) the parameter values $0<K\leq 1/3$ can be
handled in a uniform manner, and (iii)  it relies on a conformal transformation that compactifies the time domain to a finite
interval, which we expect will be useful for numerical simulating solutions globally to the future.
A further advantage of our method is that it can be easily adapted to handle
the case $K=0$; we will
present the details in a separate article.

Our approach to solving the Einstein-Euler equations \eqref{EE.1}-\eqref{EE.2} shares much
in common with the methods employed in \cite{Ringstrom:2008,RodnianskiSpeck:2013,Speck:2012} in that it is
based on a particular choice of
wave gauge that brings the Einstein-Euler system into a form that is suitable for
analyzing behaviour in the far future via energy estimates. We remark that, for us,
a suitable form means a symmetric hyperbolic system of the type analyzed in Appendix \ref{Hyp}. The key
difference between our approach and \cite{Ringstrom:2008,RodnianskiSpeck:2013,Speck:2012} is that we do not analyze the physical spacetime
metric directly, but instead use a conformal metric as our primary gravitational variable. It is
worth noting that our method does not rely on the conformal field equations of Friedrich \cite{Friedrich:1986,Friedrich:1991}.

In order to describe our approach, we first must fix our notation. The spacetimes that we consider
are of the form $M=(0,1]\times \Tbb^3$. We use $(x^I)$, $(I=1,2,3)$, to denote (spatial) periodic coordinates
on $\Tbb^3$, and $t=x^0$ to denote a time coordinate on the interval $(0,1]$ where the future lies in the direction
of \emph{decreasing} $t$. Unless stated otherwise, we use lower case Latin letters, i.e. $i,j,k$, to index
spacetime coordinate indices that run from $0$ to $3$ while upper case Latin letters, i.e. $I,J,K$, will
index spatial coordinate indices that run from $1$ to $3$. Partial derivatives with respect to the
coordinates $(x^i)$ will be denoted by
$\del{i}=\del{}/\del{}x^i$. The fluid-four velocity $\vt^i$ is assumed to be future oriented, which is
equivalent to the condition
\leqn{vtfuture}{
\vt^0 < 0.
}

As indicated above, we do not work with the spacetime metric \eqref{gtdef} directly, but instead use
the conformally transformed metric
\leqn{ctrans}{
g_{ij} = e^{-2\Phi}\gt_{ij}.
}
As has been demonstrated previously in \cite{Friedrich:1986,Friedrich:1991}, a key technical advantage of the conformal approach is
that it turns global existence problems into local ones via the compactification of spacetime. In our setting,
this amounts to us being able to obtain solutions globally to the future by solving the conformal Einstein-Euler
equations on a finite time interval. An additional benefit of the conformal approach is that it allows us to account for the leading order
effects of accelerated expansion through a judicious choice of the conformal factor.

Under the conformal transformation \eqref{ctrans},
the Einstein equations \eqref{EE.1} transform as
\eqn{cEineqnsG}{
G^{ij} = T^{ij} := e^{4\Phi}\Tt^{ij}-e^{2\Phi}\Lambda g^{ij} + 2\bigl(\nabla^i\nabla^j\Phi-\nabla^i\Phi\nabla^j\Phi\bigr)
-\bigl(2\Box \Phi + |\nabla\Phi|_g^2\bigr)g^{ij},
}
or equivalently
\leqn{cEineqns}{
-2R_{ij} =
- 4 \nabla_i\nabla_j \Phi + 4 \nabla_i \Phi \nabla_j\Phi - 2\biggl[ \Box \Phi + 2|\nabla\Phi|^2 + \biggl(\frac{1-K}{2}\rho+\Lambda\biggr)e^{2\Phi} \biggr] g_{ij} -2 e^{2\Phi}(1+K)\rho v_i v_j,
}
where we employ the notation $\Box = g^{ij}\nabla_i\nabla_j$ and $|\nabla\Phi|^2 = g^{ij}\nabla_i\Phi\nabla_j \Phi$ and
we have introduced the \emph{conformal fluid four-velocity}
\eqn{vtdef}{
v^i = e^{\Phi}\vt^i.
}
We note that the four-velocity future orientation condition \eqref{vtfuture} is equivalent to
\eqn{vfuture}{
v_0 > 0
}
where $v_0 = g_{0j}v^j$. Unless otherwise specified, we raise and lower all coordinate tensor indices
using the conformal metric $g_{ij}$ and its inverse $g^{ij}$ as appropriate.

Letting $\Gammat{}^k_{ij}$ and $\Gamma^k_{ij}$ denote the Christoffel symbols of the metrics $\gt_{ij}$
and $g_{ij}$, respectively, the difference $\Gammat^k_{ij}-\Gamma^k_{ij}$ is readily calculated to be
\eqn{Gamdiff}{
\Gammat{}^k_{ij}-\Gamma^k_{ij} = g^{kl}\bigl(g_{il}\nabla_j\Phi
+ g_{jl}\nabla_i\Phi - g_{ij} \nabla_l\Phi \bigr).
}
Using this, we can express the Euler equations \eqref{EE.2} as
\leqn{ceul}{
\nabla_i \Tt^{ij} = -6\Tt^{ij}\nabla_i\Phi +g_{lm}\Tt^{lm}
g^{ij}\nabla_i\Phi.
}

The wave gauge we employ is defined by the vanishing of the vector field
\leqn{Zdef}{
Z^k = X^k + Y^k
}
where
\gath{XYdef}{
X^k  := g^{ij}\Gamma_{ij}^k = -\del{i}g^{ki}+\Half g^{kl}g_{ij}\del{l}g^{ij}, \\
Y^k = - 2\nabla^k\Phi + \mu \delta^k_0, 
}
and $\mu$ is to be specified.

For the conformal factor, we choose
\leqn{Phiset}{
\Phi = -\ln(t)
}
while for $\mu$, we choose
\eqn{muset}{
\mu = \frac{2\Lambda}{3t}.
}
With these choices the vector field \eqref{Zdef} becomes
\leqn{Zset}{
Z^k = X^k +  \frac{2}{t}\biggl(g^{k0} + \frac{\Lambda}{3}\delta^{k}_0\biggr).
}
\begin{rem} \label{confrem}
Since we want the conformal factor $e^{2\Phi}$ to become unbounded towards the future to
account for accelerated expansion, the choice \eqref{Phiset} for our conformal factor
explains our choice of time orientation on the interval $(0,1]$ and shows that
future timelike infinity is located at $t=0$.
\end{rem}
For use below, we  introduce the background Minkowski metric
\eqn{Mindef}{
\eta^{ij} = -\frac{\Lambda}{3}\delta^i_0\delta^j_0 + \delta^i_I\delta^j_J \delta^{IJ},
}
the densitized three-metric
\leqn{gfdef}{
\gf^{IJ} = \det(\gch_{KL})^{\frac{1}{3}} g^{IJ}
}
where
\eqn{gchdef}{
(\gch_{IJ} ) = (g^{IJ})^{-1},
}
and the variable
\leqn{qfdef}{
\qf = g^{00}-\eta^{00}-\frac{\Lambda}{9}\ln\bigl(\det(g^{PQ})\bigr).
}

\begin{rem} \label{confvacrem}
While our choices for the coordinates and conformal factor are ultimately justified
by the fact that they work, motivation for these choices can be seen by observing that
\eqn{confvacrem1}{
(g^{ij},\Phi) = (\eta^{ij},-\ln(t))
}
is an exact solution to the vacuum conformal Einstein equations, i.e.
\eqref{cEineqns} with $\rho=0$,
that is a member of the FLRW family of solutions and
locally conformal to de Sitter space.
\end{rem}

With our notation and conventions fixed, we are now able to state our future non-linear stability result in the following
theorem,
which is the main result of this article. The proof is given in Section \ref{mproof}.
\begin{thm} \label{mthm}
Suppose $\Lambda>0$, $\epsilon > 0$, $0<K\leq 1/3$, $k\in \Zbb_{\geq 3}$, $g^{ij}_0 \in H^{k+1}(\Tbb^3)$,
$\; g^{ij}_1,\rho_0,v_I^0 \in H^{k}(\Tbb^3)$, $\rho_0(x)>0$ for all $x\in \Tbb^3$, and
that the quadruple
\leqn{mthm1}{
\bigr(g^{ij},\del{t}g^{ij},\rho,v_I\bigl)\bigl|_{t=1} = \bigr(g^{ij}_0,g^{ij}_1,\rho_0,v^0_I\bigr)
}
satisfies the constraint equations
\leqn{mthm2}{
(G^{i0}-T^{i0})|_{t=1} = 0, \AND Z^k|_{t=1} = 0.
}
Then there exists a $\delta > 0$, such that if
\leqn{mthm3}{
\norm{g^{ij}_0-\eta^{ij}}_{H^{k+1}}+\norm{g^{ij}_1}_{H^k} + \norm{\rho_0}_{H^k} + \norm{v_I}_{H^k} < \delta,
}
then there exists a classical solution $g^{ij}\in C^2((0,1]\times \Tbb^3)$, $\rho,v_I\in C^1((0,1]\times \Tbb^3)$
to the conformal Einstein-Euler equations given by \eqref{cEineqns} and \eqref{ceul} that satisfies the initial conditions \eqref{mthm1},
the gauge condition $Z^k=0$ in $(0,1]\times \Tbb^3$, the regularity conditions
\gath{mthm4}{
g^{ij} \in C^0((0,1],H^{k+1}(\Tbb^3)) \cap C^0([0,1],H^{k}(\Tbb^3))\cap  C^1\bigl((0,1],H^{k}(\Tbb^3)\bigl) \cap C^1\bigl([0,1],H^{k-1}(\Tbb^3)\bigl)
\intertext{and}
\rho, v_I \in  C^0\bigl((0,1],H^{k}(\Tbb^3)\bigr)\cap C^0\bigl([0,1],H^{k-1}(\Tbb^3)\bigr),
}
and the bounds
\gath{mthm5}{
\frac{\Lambda}{6} \leq -g^{00}(t,x) \leq \frac{2\Lambda}{3}, \quad \frac{3}{2\Lambda} \leq -g_{00}(t,x) \leq \frac{6}{\Lambda},
\quad
\frac{1}{2}\delta^{IJ}  \leq g^{IJ}(t,x) \leq  \frac{3}{2}\delta^{IJ}, \\
\intertext{and}
\sqrt{\frac{3}{2\Lambda}} \leq v_0(t,x) \leq \sqrt{\frac{6}{\Lambda}}
}
for all $(t,x)\in (0,1]\times \Tbb^3$.

Moreover, there exist $\sigma, \gamma^j\in H^{k-1}(\Tbb^3)$ with $\sigma(x) > 0$ for all $x\in \Tbb^3$,
such that the solution satisfies
\lalin{mthm6}{
\norm{g^{0j}(t)-(\eta^{0j}+t\gamma^j)}_{H^{k-1}} &\lesssim t^{2-\epsilon}, \label{mthm6.1} \\
\norm{\del{t}g^{0j}(t)-2t^{-1}(g^{0j}(t)-\eta^{0j}) + \gamma^j}_{H^{k-1}} &\lesssim t, \label{mthm6.2} \\
\norm{\del{t}g^{0j}(t)-t^{-1}(g^{0j}(t)-\eta^{0j})}_{H^{k-1}} + \norm{\del{I}g^{0j}(t)}_{H^{k-1}} &\lesssim t^{1-\epsilon},
\label{mthm6.3} \\
\norm{\qf(t)- \qf(0)}_{H^{k}} + \norm{\del{t}\qf(t)}_{H^{k-1}} &\lesssim t, \label{mthm6.4} \\
\norm{\gf^{IJ}(t)-\gf^{IJ}(0)}_{H^{k}} + \norm{\del{t}\gf^{IJ}(t)}_{H^{k-1}} &\lesssim t,\label{mthm6.5} \\
\norm{t^{-3(1+K)}\rho(t)-\sigma}_{H^{k-1}} &\lesssim t+t^{\frac{2(1-3K)}{(1+\epsilon)}}, \label{mthm6.6}
\intertext{and}
\norm{v_I(t)-v_I(0)}_{H^{k-1}} &\lesssim  t^{\frac{1-3K}{(1+\epsilon)}} \label{mthm6.7}
}
for all $t\in [0,1]$, where $v_I(0)=0$ if $K\neq 1/3$.
\end{thm}
We conclude with a few remarks:
\begin{enumerate}[(i)]
\item For any specified $\delta >0$, an open set of initial data satisfying the constraint equations \eqref{mthm2} and
 the condition \eqref{mthm3} can be constructed
using a variation of the method employed in \cite[\S 3]{Oliynyk:CMP_2010}.
\item The FLRW family of solutions correspond to the solutions obtained from initial data \eqref{mthm1} of
the form
\eqn{FLRWidata}{
\bigr(g^{ij}_0,g^{ij}_1,\rho_0,v^0_I\bigr)=
\bigl(-a_0\delta^{i}_0\delta^{j}_0+b_0\delta^i_I\delta^j_J\delta^{IJ},
-a_1\delta^{i}_0\delta^{j}_0+b_1\delta^i_I\delta^j_J\delta^{IJ},c,0\bigr)
}
where the constants $a_0,b_0,c\in \Rbb_{>0}$ and $a_1,b_1\in \Rbb$ are chosen so that the constraint equations \eqref{mthm2} are
satisfied.
\item More detailed behaviour of the solution near $t=0$ can be obtained by using the estimates
\eqref{mthm6.1}-\eqref{mthm6.7} from Theorem \ref{mthm}
together with another application of Theorem \ref{symthm} from Appendix \ref{Hyp}.
\item As discussed in Section 1.3 of \cite{RodnianskiSpeck:2013}, it should be possible using Ringstr\"{o}m's patching argument from \cite{Ringstrom:2008} to generalize the
stability result contained in Theorem \ref{mthm} to other spatial topologies.
\end{enumerate}



\sect{cEin}{The reduced conformal Einstein-Euler equations}

As discussed briefly in the introduction, our proof of Theorem \ref{mthm}
relies on expressing the conformal Einstein-Euler equations, given
by \eqref{cEineqns} and \eqref{ceul}, as a symmetric hyperbolic system of the type analyzed in Appendix \ref{Hyp}.
We achieve this through the use of a wave gauge defined by the vanishing
of the vector field \eqref{Zdef} in conjunction with a suitable transformation of the field variables.
We present the details of this procedure in Sections \ref{redconf} and \ref{ceuler}.

\subsect{redconf}{The reduced conformal Einstein equations}

Following H. Friedrich's \cite{Friedrich:1985} generalization of the gauge reduction method of Y.Choquet-Bruhat \cite[Ch. VI, \S 8]{ChoquetBruhat:2009} that, crucially, allows for gauge source functions,
we replace the Einstein equations \eqref{cEineqns} by
the gauge reduced version:
\alin{cEineqnsA}{
-2R_{ij}+ 2\nabla_{(i}Z_{j)}+&A_{ij}{}^k Z_k = - 4 \nabla_i\nabla_j \Phi + 4 \nabla_i \Phi \nabla_j\Phi \notag\\
&- 2\biggl[ \Box \Phi + 2|\nabla\Phi|^2 +\biggl(\frac{1-K}{2}\rho+\Lambda\biggr)e^{2\Phi} \biggr] g_{ij} -2(1+K)e^{2\Phi} \rho v_i v_j 
}
where
\eqn{Adef}{
A^{ij}{}_k = -X^{(i}\delta^{j)}_k+ Y^{(i}\delta^{j)}_k.
}
We will refer to these equations as the \emph{reduced conformal Einstein equations}.

Since the gauge condition $Z^k=0$ propagates, we can, with out loss of generality, assume that
it holds, or equivalently that
\eqn{ZzeroA}{
X^k = -Y^k,
}
which, in turn, implies that
\leqn{ZzeroB}{
\nabla_i \nabla^i t = Y^0 = \frac{2}{t}\biggl(g^{00}+\frac{\Lambda}{3}\biggr).
}

A straightforward calculation using \eqref{Phiset}, \eqref{Zset}, \eqref{ZzeroB} and the identity
\eqn{symDdeltaj0}{
\nabla^{(i}\chi^{j)} = -\frac{1}{2}\del{t}g^{ij} \qquad (\chi^j=\delta^j_0)
}
then shows that reduced Einstein equations are equivalent to
\lalin{cEineqnsB}{
-2R^{ij} + 2\nabla^{(i}X^{j)} -X^i X^j = \frac{2\Lambda}{3t} \del{t}g^{ij} -\frac{4\Lambda}{3t^2}&\biggl(g^{00}+\frac{\Lambda}{3}\biggr)\delta^{i}_0\delta^{j}_0
-\frac{4\Lambda}{3t^2}g^{0K}\delta^{(i}_{0}\delta^{j)}_{K} \notag \\
&  -\frac{2}{t^2}g^{ij}\biggl(g^{00}+\frac{\Lambda}{3}\biggr) - (1-K)\frac{\rho}{t^2}g^{ij} - 2(1+K)\frac{\rho}{t^2} v^i v^j, \label{cEineqnsB.1}
}
which we note can be expressed as
\lalin{cEineqnsC}{
-g^{kl}\del{k}\del{l}g^{ij}  - Q^{ij}(g,\del{}g) = \frac{2\Lambda}{3t} \del{t}g^{ij} -\frac{4\Lambda}{3t^2}&\biggl(g^{00}+\frac{\Lambda}{3}\biggr)\delta^{i}_0\delta^{j}_0
-\frac{4\Lambda}{3t^2}g^{0K}\delta^{(i}_{0}\delta^{j)}_{K} \notag \\
&  -\frac{2}{t^2}g^{ij}\biggl(g^{00}+\frac{\Lambda}{3}\biggr) - (1-K)\frac{\rho}{t^2}g^{ij} - 2(1+K)\frac{\rho}{t^2} v^i v^j \label{cEineqnsC.1}
}
where the $Q^{ij}$ are quadratic in $\del{}g=(\del{k}g^{ij})$ and analytic in $g=(g^{ij})$ on the region defined by $\det(g^{ij})<0$.

The first step in transforming \eqref{cEineqnsC.1} into the desired form involves
parameterizing the three-metric $g^{IJ}$ using densitized spatial metric \eqref{gfdef} and the
determinant variable \eqref{qfdef}.
Using the identities
\eqn{LidA}{
-g^{kl}\del{k}\del{l}\bigl(\gf^{IJ}\bigr) = \det(\gch_{PQ})^{\frac{1}{3}}\Lc^{IJ}_{LM}\bigl(-g^{lk}\del{l}\del{k} g^{LM}\bigr)
-g^{kl}\del{k}\bigl(\det(\gch_{PQ})^{\frac{1}{3}} \Lc^{IJ}_{LM}\bigr)   \del{l}g^{LM}
}
and
\eqn{LidB}{
\del{t}\gf^{IJ} = \det(\gch_{PQ})^{\frac{1}{3}} \Lc^{IJ}_{LM}\del{t} g^{LM},
}
where
\eqn{Ldef}{
\Lc^{IJ}_{LM} = \delta^I_L\delta^J_M - \Third \gch_{LM} g^{IJ},
}
we see from \eqref{cEineqnsC.1} that $\gf^{IJ}$ and $\qf$ satisfy
\leqn{cEineqnsD}{
-g^{kl}\del{k}\del{l}\gf^{IJ}  - \Qt^{IJ}(g,\del{}g) = \frac{2\Lambda}{3t} \del{t}\gf^{IJ}
 - 2(1+K)\det(\gch_{PQ})^{\frac{1}{3}} \Lc^{IJ}_{LM}\frac{\rho}{t^2} v^L v^M
}
and
\leqn{cEineqnsE}{
-g^{kl}\del{k}\del{l}\qf  - \Qt(g,\del{}g) = \frac{2\Lambda}{3t} \del{t}\qf
- 2\biggl(\frac{g^{00}+\frac{\Lambda}{3}}{t}\biggr)^2+(1-K)\biggl(\frac{\Lambda}{3}-g^{00}\biggr)\frac{\rho}{t^2} + 2(1+K)\biggl(\frac{2\Lambda}{9} \gch_{IJ}v^I v^J - v^{0}v^{0}\biggr)\frac{\rho}{t^2},
}
respectively,
where $\Qt^{IJ}$ and $\Qt$ are quadratic in $\del{}g=(\del{k}g^{ij})$ and analytic in $g=(g^{ij})$ on the region defined by $\det(g^{IJ})>0$ and
$g_{00} < 0$.

We proceed by writing the wave equations \eqref{cEineqnsC.1} $(i=0)$, \eqref{cEineqnsD} and \eqref{cEineqnsE} in first order form using the following definitions:
\lalin{ufdef}{
\uf^{0j} &= \frac{g^{0j}-\eta^{0j}}{2t}, \label{ufdef.1} \\
\uf^{0j}_0 &= \del{t}g^{0j}- \frac{3(g^{0j}-\eta^{0j})}{2t}, \label{ufdef.2} \\
\uf_I^{0j} &= \del{I}g^{0j}, \label{ufdef.3} \\
\uf^{IJ} &= \gf^{IJ} - \delta^{IJ}, \label{ufdef.4} \\
\uf^{IJ}_k &= \del{k}\gf^{IJ}, \label{ufdef.5}\\
\uf &= \qf, \label{ufdef.6} \\
\intertext{and}
\uf_k &= \del{k}\qf. \label{ufdef.7}
}
In terms of theses first order variables, it is not difficult to verify that the wave equations \eqref{cEineqnsC.1} $(i=0)$, \eqref{cEineqnsD}, and \eqref{cEineqnsE} are equivalent
to the following systems of symmetric hyperbolic equations:
\lalin{cEineqnsF}{
A^k\del{k}\begin{pmatrix}\uf^{0l}_0  \\ \uf^{0l}_J \\ \uf^{0l} \end{pmatrix} &= \frac{1}{t}\Ac \Pbb \begin{pmatrix}\uf^{0l}_0  \\ \uf^{0l}_J \\ \uf^{0l}
\end{pmatrix}  + F, \label{cEineqnsF.1}\\
A^k\del{k}\begin{pmatrix}\uf^{LM}_0  \\ \uf^{LM}_J \\ \uf^{LM} \end{pmatrix} &= -\frac{2g^{00}}{t}
\Pi\begin{pmatrix}\uf^{LM}_0  \\ \uf^{LM}_J  \\ \uf^{LM}
\end{pmatrix} + \Ff, \label{cEineqnsF.2}
\intertext{and}
A^k\del{k}\begin{pmatrix}\uf_0  \\ \uf_J \\ \uf \end{pmatrix} &= -\frac{2g^{00}}{t}
\Pi\begin{pmatrix}\uf_0  \\ \uf_J \\ \uf
\end{pmatrix} + \ff \label{cEineqnsF.3}
}
where
\lalin{Akdef}{
A^0 &= \begin{pmatrix} -g^{00} & 0 & 0 \\ 0& g^{IJ} & 0 \\ 0 & 0 & -g^{00} \end{pmatrix}, \label{Akdef.1} \\
A^K &= \begin{pmatrix}  - 2g^{0K} & -g^{JK} & 0 \\ -g^{IK} & 0 & 0  \\ 0 & 0 & 0 \end{pmatrix},
\label{Akdef.2} \\
\Pbb & = \begin{pmatrix} \frac{1}{2} & 0 & \frac{1}{2} \\ 0 & \delta^J_K & 0 \\
\frac{1}{2} & 0 & \frac{1}{2}  \end{pmatrix}, \label{Akdef.3}\\
\Ac &= \begin{pmatrix} -g^{00} & 0 & 0 \\ 0 & \frac{3}{2} g^{IK} & 0 \\
0 & 0 & -g^{00}  \end{pmatrix}, \label{Akdef.4}\\
\Pi & = \begin{pmatrix} 1 & 0 & 0 \\ 0 & 0 & 0 \\
0 & 0 & 0  \end{pmatrix}, \label{Akdef.5}\\
F &= \begin{pmatrix}Q^{0l} + 4\uf^{00}(\uf^{0l}_0+\uf^{0l})-8\uf^{00}\uf^{0l} + 6 \uf^{0I}\uf^{0l}_I -(1-K)\frac{\rho}{t^2}g^{0l}-2(1+K)\frac{\rho}{t^2} v^0 v^l
\\ 0 \\ 0\end{pmatrix}, \label{Akdef.6}\\
\Ff &= \begin{pmatrix} \Qt^{LM}+4\uf^{00}\uf^{LM}_0 - 2(1+K)\det(\gch_{PQ})^{\frac{1}{3}}\Lc^{LM}_{RS}\frac{\rho}{t^2} v^R v^S \\ 0 \\ -g^{00} \uf^{LM}_0  \end{pmatrix},  \label{Akdef.7}
\intertext{and}
\ff &= \begin{pmatrix} \Qt+4\uf^{00}\uf_0 -8(\uf^{00})^2 + (1-K)\Bigl(\frac{\Lambda}{3}-g^{00}\Bigr)\frac{\rho}{t^2} +
2(1+K)\Bigl(\frac{2\Lambda}{9}\gch_{RS}v^R v^S - v^0 v^0\Bigr)\frac{\rho}{t^2} \\ 0 \\ -g^{00} \uf^{LM}_0 \end{pmatrix}.
\label{Akdef.8}
}

It is important to note that the term $\frac{\rho}{t^2}$ that appears
in \eqref{Akdef.6}-\eqref{Akdef.8} is not a singular term. This is because we can bound the density variable $\zeta$ defined by \eqref{ceulB.1} in the following section. Expressing $\frac{\rho}{t^2}$ in terms of $\zeta$, it is clear that this
term is regular.

Another important point is that the first order variables \eqref{ufdef.1}-\eqref{ufdef.7} are completely equivalent for $t>0$ to the metric $g^{ij}$ and its first derivatives $\del{k}g^{ij}$.
Consequently, we lose nothing by parameterizing the gravitational field in this fashion, and gain by having transformed the reduced conformal equations into a symmetric hyperbolic system of the form\footnote{To be precise, we still need to make the trivial coordinate transformation $t \mapsto -t$ to obtain the form considered in Appendix \ref{Hyp}.} considered in Appendix \ref{Hyp}.

\subsect{ceuler}{The conformal Euler equations}
For the conformal Euler equations \eqref{ceul}, we use the formulation from \cite[\S 2.2]{Oliynyk:CMP_2015}; see, in particular, equation (2.55) of that article. In this formulation, given the
conformal factor \eqref{Phiset} and the linear equation of state \eqref{eosdef},
the conformal Euler equations \eqref{ceul} are from the computations presented in \cite[\S 2.2]{Oliynyk:CMP_2015} readily seen
to be given by
\leqn{ceulA}{
B^k \del{k}\begin{pmatrix} \zeta \\ v_J \end{pmatrix}  = \frac{1}{t}\Bc \pi\begin{pmatrix} \zeta \\ v_J \end{pmatrix}  -  H
}
where
\lalin{ceulB}{
\rho &= t^{3(1+K)} \rho_c e^{(1+K)\zeta}, \qquad \rho_c \in \Rbb_{> 0}, \label{ceulB.1}\\
v_0 & = v_0(g^{kl},v_L) := \frac{-g^{0I}v_I - \sqrt{(g^{0I}v_I)^2-g^{00}(g^{IJ}v_I v_J +1 )} }{g^{00}} ,  \label{ceulB.2} \\
v^i & =  v^i(g^{kl},v_L) := g^{iJ}v_J + g^{i0}v_0(g^{kl},v_L) , \label{ceulB.3} \\
V^{IJ} &:= \frac{\del{}v^I}{\del{}v_J} = g^{IJ} - g^{I0}\frac{v^J}{v^0}, \label{ceulB.4}\\
\Bc &= \frac{-1}{v^0}\begin{pmatrix} 1 & 0 \\ 0 & \frac{1-3K}{v_0}V^{JI} \end{pmatrix}, \label{ceulB.5} \\
\pi &= \begin{pmatrix} 0 & 0 \\ 0 & \delta_{I}^J \end{pmatrix},\label{ceulB.6} \\
L^k_I &= \delta^k_J - \frac{v_J}{v_0} \delta^k_0, \label{ceulB.7} \\
M_{IJ} &= g_{IJ} - \frac{v_I}{v_0}g_{0J} - \frac{v_{J}}{v_0}
g_{I0} + \frac{g_{00}}{(v_0)^2}v_I v_J, \label{ceulB.8}\\
B^0 &= \begin{pmatrix} K  &  \frac{K}{v^0}  L^0_M V^{MJ} \\ \frac{K}{v^0} V^{LI} L^0_L & V^{LI} M_{LM} V^{MJ} \end{pmatrix}, \label{ceulB.9}\\
B^K &= \frac{1}{v^0}\begin{pmatrix} Kv^K  &  K  L^K_M V^{MJ} \\ K V^{LI} L^K_L & V^{LI} M_{LM} V^{MJ} v^K \end{pmatrix}, \label{ceulB.9a}
\intertext{and}
H & = -\frac{1}{v^0}\begin{pmatrix}
-K L^k_L \Gamma^L_{kl}v^l - KL^k_L \frac{\del{} v^L}{\del{}g^{lm}}\del{k} g^{lm} \\
- v^{LI}M_{LM}v^k \Gamma_{kl}^M v^l - V^{LI}M_{LM}v^k  \frac{\del{} v^L}{\del{}g^{lm}}\del{k} g^{lm}
\end{pmatrix}. \label{ceulB.10}
}

\sect{mproof}{Proof of Theorem \ref{mthm}}
We begin by fixing $\Lambda>0$, $k\in \Zbb_{\geq 3}$, $K\in (0,1/3]$, $\epsilon > 0$ and $R>0$. Assuming $g_{00}<0$, $g^{00}<0$ and $\det(g^{IJ}) >0$, we then observe that the relations
\leqn{mproof1}{
\Pbb A^0 \Pbb \leq \Pbb \Ac \Pbb \AND  \Pi A^0 \Pi = \frac{1}{2}\bigl(-2g^{00}\Pi\bigr)
}
follow directly from \eqref{Akdef.1}, \eqref{Akdef.3} and \eqref{Akdef.5}.

Next, setting
\eqn{Audef}{\Au^0 = \begin{pmatrix} A^0 & 0 & 0 \\ 0  & A^0 & 0 \\
0 & 0 & A^0  \end{pmatrix},
}
\eqn{Acudef}{\Acu = \begin{pmatrix} \Ac & 0 & 0 \\ 0  & -2g^{00}\id & 0 \\
0 & 0 & -2 g^{00}\id  \end{pmatrix},
}
\eqn{Pbbudef}{
\Pbbu = \begin{pmatrix} \Pbb & 0 & 0 \\ 0  & \Pi & 0 \\
0 & 0 & \Pi  \end{pmatrix},
}
and
\eqn{ufudef}{
\ufu = \bigl(\uf^{0l},\uf^{0l}_J,\uf^{0l},\uf^{LM}_0,\uf^{LM}_J,\uf^{LM},\uf_0,\uf_J,\uf\bigr)^{\tr},
}
a straightforward calculation shows that
\leqn{Acom}{
[\Au^0,\Pbbu]=[\Acu,\Pbbu] = 0,
}
while we see from \eqref{Akdef.1} that
\leqn{Auvar}{
\Au^0 = \Au^0(t,t\ufu,\Pbbu^{\perp}\ufu).
}
Using \eqref{Acom} and \eqref{Auvar}, we compute
\leqn{DAuvar}{
D_{\ufu}\Au^0\cdot (\Au^0)^{-1}\Acu\Pbbu \ufu = t \bigl[ D_2 \Au^0(t,t\ufu,\Pbbu^{\perp}\ufu)\cdot (\Au^0)^{-1}\Acu\Pbbu \ufu\bigr] \qquad 
(\Pbbu^{\perp} := \id - \Pbbu).
}

Next, it is clear from \eqref{ceulB.5} and \eqref{ceulB.6} that
\leqn{Bccom}{
[\Bc,\pi] = 0,
}
and that
\leqn{mproof1a}{
\pi\Bc = 0, \qquad \text{if $K=1/3$.}
}
Setting
\eqn{vfudef}{
\xi = (\ufu,\zeta,v_J), 
}
and evaluating $B^0$ and $\Bc$ at $\xi=0$, we find that
\eqn{Bczero}{
B^0|_{\xi=0} = \begin{pmatrix} K & 0 \\ 0 & \delta_{IJ} \end{pmatrix}
\AND
\Bc|_{\xi=0} = \begin{pmatrix} \sqrt{\frac{3}{\Lambda}} & 0 \\ 0 & 1-3K \end{pmatrix}.
}
From these expressions, it is clear that
\leqn{mproof3a}{
\pi B^0|_{\xi=0} \pi  = \frac{1}{1-3K}\pi \Bc|_{\xi=0}=(\delta^{IJ}), \qquad \text{if $K\in (0,1/3)$}.
}
We also observe from \eqref{ceulB.9} that
\eqn{B0vars}{
\pi^\perp B^0(t,\xi)\pi^\perp = K \qquad (\pi^\perp := \id - \pi),
}
which, in turn, implies after differentiating that
\leqn{B0diff}{
\pi^\perp \bigl[D_{\xi} B^0(t,\xi)\cdot \delta\xi\bigr] \pi^\perp  = 0
}
for all $(t,\xi)$ in the domain of definition of $B^0$ and all variations $\delta\xi$.

From the evolution equations \eqref{cEineqnsF.1}, \eqref{cEineqnsF.2}, \eqref{cEineqnsF.3},
and \eqref{ceulA}, the relations \eqref{mproof1}, \eqref{Acom}, \eqref{DAuvar}, \eqref{Bccom}, \eqref{mproof1a}, \eqref{mproof3a}, and \eqref{B0diff}, and
Remark \ref{boundrem}, it is clear after performing
the trivial time-coordinate transformation $t \mapsto -t$ that $\xi$ satisfies a symmetric hyperbolic
system of the type to which we can apply Theorem \ref{symthm} from Appendix \ref{Hyp}. Doing so guarantees
the existence of a $\delta >0$ such that if
\leqn{mproof4}{
(g^{ij},\del{t}g^{ij},\zeta,v_J)\bigl|_{t=1} =  (g^{ij}_0,\del{t}g^{ij}_0,\zeta_0,v_J^0)\bigl|_{t=1}\in H^{k+1}(\Tbb^3)\times H^{k}(\Tbb^3) \times H^{k}(\Tbb^3) \times H^{k}(\Tbb^3)
}
and
\eqn{mproof4a}{
\rho_c \in \Rbb_{>0}
}
is chosen so that $\rho_c$ and the corresponding initial data $\xi|_{t=1}=\xi_0$, which can be computed via the formulas \eqref{gfdef}, \eqref{qfdef}, \eqref{ufdef.1}-\eqref{ufdef.7} and \eqref{ceulB.1}, are bounded by
\eqn{mproof5}{
\norm{\xi_0}_{H^k} + \rho_c \leq \delta,
}
then there exists a map
\eqn{mproof6}{
\xi \in C^1\bigl((0,1]\times \Tbb^3]\bigr)\cap C^0\bigl((0,1],H^{k}(\Tbb^3)\bigr) \cap C^0\bigl([0,1],H^{k-1}(\Tbb^3)\bigr)
}
such that
\eqn{mproof7}{
\norm{\xi}_{L^\infty([0,1],H^{k})} \leq R,
}
$\xi$ is the unique classical solution to \eqref{cEineqnsF.1}, \eqref{cEineqnsF.2}, \eqref{cEineqnsF.3}
and \eqref{ceulA} on $(0,1]\times \Tbb^3$
that agrees with the initial data at $t=1$, i.e. $\xi|_{t=1}=\xi_0$,
the metric $g^{ij}$ and conformal four-velocity component $v_0$, as determined by the formulas \eqref{gfdef}, \eqref{qfdef}, \eqref{ufdef.1}, \eqref{ufdef.4}, \eqref{ufdef.6} and
\eqref{ceulB.2}, satisfy
\gath{mproof8a}{
\frac{\Lambda}{6} \leq -g^{00}(t,x) \leq \frac{2\Lambda}{3}, \quad \frac{3}{2\Lambda} \leq -g_{00}(t,x) \leq \frac{6}{\Lambda},
\quad
\frac{1}{2}\delta^{IJ}  \leq g^{IJ}(t,x) \leq  \frac{3}{2}\delta^{IJ}, \\
\intertext{and}
\sqrt{\frac{3}{2\Lambda}} \leq v_0(t,x) \leq \sqrt{\frac{6}{\Lambda}}
}
for all $(t,x)\in (0,1]\times \Tbb^3$. Moreover, for given $\epsilon >0$,  it follows, by choosing $R>0$ small enough, that the components of $\xi$ satisfy:
\lalin{mproof8}{
\norm{\uf^{0j}_0(t)-\uf^{0j}(t)+2\uf^{0j}(0)}_{H^{k-1}} &\lesssim t, \label{mproof8.1}\\
\norm{\uf^{0j}_0(t)+\uf^{0j}(t)}_{H^{k-1}} + \norm{\uf_I^{0j}}_{H^{k-1}} &\lesssim t^{1-\epsilon}, \label{mproof8.2}\\
\norm{\uf^{IJ}(t)-\uf^{IJ}(0)}_{H^{k-1}} + \norm{\uf^{IJ}_K(t)-\uf^{IJ}_K(0)}_{H^{k-1}} + \norm{\uf^{IJ}_0(t)}_{H^{k-1}} &\lesssim t, \label{mproof8.3}\\
\norm{\uf(t)-\uf(0)}_{H^{k-1}}+ \norm{\uf_I(t)-\uf_I(0)}_{H^{k-1}}  + \norm{\del{t}\uf(t)}_{H^{k-1}} &\lesssim t, \label{mproof8.4}\\
\norm{\zeta(t)-\zeta(0)}_{H^{k-1}} &\lesssim t+ t^{\frac{2(1-3K)}{(1+\epsilon)}}, \label{mproof8.5}
\intertext{and}
\norm{v_I(t)-v_I(0)}_{H^{k-1}} &\lesssim  t^{\frac{1-3K}{(1+\epsilon)}} \label{mproof8.6}
}
for $0\leq t \leq 1$, where
\eqn{mproof9}{
v_I(0)=0 \quad \text{if $K\in (0,1/3)$.}
}
Using \eqref{mproof8.1}-\eqref{mproof8.6}, it follows directly from the formulas \eqref{gfdef}, \eqref{qfdef}, \eqref{ufdef.1}-\eqref{ufdef.7}, and \eqref{ceulB.1} that the metric and fluid density satisfy:
\alin{mproof11}{
\norm{g^{0j}(t)-(\eta^{0j}+t\gamma^j)}_{H^{k-1}} &\lesssim t^{2-\epsilon}, \\
\norm{\del{t}g^{0j}(t)-2t^{-1}(g^{0j}(t)-\eta^{0j}) + \gamma^j}_{H^{k-1}} &\lesssim t,\\
\norm{\del{t}g^{0j}(t)-t^{-1}(g^{0j}(t)-\eta^{0j})}_{H^{k-1}} + \norm{\del{I}g^{0j}}_{H^{k-1}} &\lesssim t^{1-\epsilon},\\
\norm{\qf(t)- \qf(0)}_{H^{k}} + \norm{\del{t}\qf(t)}_{H^{k-1}} &\lesssim t,\\
\norm{\gf^{IJ}(t)-\gf^{IJ}(0)}_{H^{k}} + \norm{\del{t}\gf^{IJ}(t)}_{H^{k-1}} &\lesssim t,
\intertext{and}
\norm{t^{-3(1+K)}\rho(t)-\sigma}_{H^{k-1}} &\lesssim t+t^{\frac{2(1-3K)}{(1+\epsilon)}},
}
where $\sigma, \gamma^j\in H^{k-1}(\Tbb^3)$, and $\sigma(x)>0$ for all $x\in \Tbb^3$.

Since $\xi$ satisfies \eqref{cEineqnsF.1}, \eqref{cEineqnsF.2}, \eqref{cEineqnsF.3}
and \eqref{ceulA} on $(0,1]\times \Tbb^3$, it is clear from the results of Section \ref{redconf} that the triple $\bigl(g^{ij},\rho=t^{3(1+K)}\rho_c e^{(1+K)\zeta},v_I\bigr)$ determined by $\xi$
solves the reduced conformal Einstein-Euler equations given by \eqref{ceul} and \eqref{cEineqnsB.1} on  $(0,1]\times \Tbb^3$. This fact together with the assumption
that the initial data \eqref{mproof4} satisfies the constraint equations \eqref{mthm2} allows us to conclude
via the Proposition on pg. 540 of \cite{Friedrich:1985} that the triple $\bigl(g^{ij},\rho=t^{3(1+K)}\rho_c e^{(1+K)\zeta},v_I\bigr)$  also
satisfies the full conformal Einstein-Euler equations, given by \eqref{cEineqns} and \eqref{ceul}, and the gauge condition $Z^k=0$ on  $(0,1]\times \Tbb^3$.
This completes the proof of Theorem \ref{mthm}.


\bigskip

\noindent \textit{Acknowledgements:} This work was partially supported by the Australian Research Council grant FT1210045. I would like to thank C. Liu for useful discussions, and, in particular, for identifying
an error in the statement and proof of a previous version of Theorem \ref{symthm}. I also thank the referees for their comments
and criticisms, which have served to improve the content
and exposition of this article.

\appendix


\sect{calc}{Calculus inequalities}
In this appendix, we collect, for the convenience of the reader, a number of calculus inequalities that we employ throughout this article. The proof of the following inequalities are well known and may be found, for example, in
the books \cite{AdamsFournier:2003}, \cite{Friedman:1976} and \cite{TaylorIII:1996}. Before stating theses calculus inequalities, we first fix our notation and introduced a few definitions that will be needed in this appendix and the
following one:

\vspace{0.15cm}

\noindent \textbf{Coordinates, indexing and derivatives:} Lower case Latin indices, e.g. $i,j,k$, will run from $0$ to $n$, while upper case
Latin indices, e.g. $I,J,K$, will run from $1$ to $n$. We use $x=(x^I)$ to denote the standard periodic coordinates
on the $n$-Torus $\Tbb^n$ and $x^0=t$ to denote the time coordinate on intervals of
the form $[T_0,T_1)$ where $T_0 < T_1 \leq 0$. Furthermore, we use lower case Greek letters to denote multi-indices, e.g.
$\alpha = (\alpha_1,\alpha_2,\ldots,\alpha_n)\in \Zbb_{\geq 0}^n$, and employ the standard notation $D^\alpha = \del{1}^{\alpha_1} \del{2}^{\alpha_2}\cdots
\del{n}^{\alpha_n}$ for spatial partial derivatives.

\vspace{0.15cm}

\noindent \textbf{Inner-products and matrix inequalities:}
We employ the notation
\eqn{ipedef}{
\ipe{\xi}{\zeta} = \xi^T \zeta, \qquad \xi,\zeta \in \Rbb^N,
}
for the Euclidean inner-product on $\Rbb^N$, and we denote the $L^2$ inner-product on $\Tbb^n$ by
\eqn{ipdef}{
\ip{u}{v} = \int_{\Tbb^n} \ipe{u(x)}{v(x)}\, d^n x.
}
For matrices $A,B\in \Mbb{N}$, we define
\eqn{Minq}{
A \leq B \quad \Longleftrightarrow \quad \ipe{\xi}{A\xi} \leq \ipe{\xi}{B\xi}, \quad \forall \: \xi \in \Rbb^N.
}

\vspace{0.15cm}

\noindent \textbf{Constants and inequalities:}
We use that standard notation
\eqn{asimb}{
a \lesssim b
}
for inequalities of the form
\eqn{aleqb}{
a \leq Cb
}
in situations where the precise value or dependence on other quantities of the constant $C$ is not required.
On the other hand, when the dependence of the constant on other inequalities needs to be specified, for
example if the constant depends on the norm $\norm{u}_{L^\infty}$, we use the notation
\eqn{cdep}{
C=C(\norm{u}_{L^\infty}).
}
Constants of this type will always be non-negative, non-decreasing, continuous functions of their arguments.

\vspace{0.15cm}

With the preliminaries out of the way, we now state the calculus inequalities:

\begin{thm}{\emph{[H\"{o}lder's inequality]}} \label{Holder}
If $0< p,q,r \leq \infty$ satisfy $1/p+1/q = 1/r$, then
\eqn{Holder1}{
\norm{uv}_{L^r} \leq \norm{u}_{L^p}\norm{v}_{L^q}
}
for all $u\in L^p(\Tbb^n)$ and $v\in L^q(\Tbb^n)$.
\end{thm}

\begin{thm}{\emph{[Sobolev's inequality]}} \label{Sobolev} Suppose
$1\leq p < \infty$ and $s\in \Zbb_{> n/p}$. Then
\eqn{Sobolev3}{
\norm{u}_{L^\infty} \lesssim \norm{u}_{W^{s,p}}
}
for all $u\in W^{s,p}(\Tbb^n)$.
\end{thm}

\begin{thm}{\emph{[Product and commutator estimates]}} \label{calcpropB} $\;$

\begin{enumerate}[(i)]
\item
Suppose $1\leq p_1,p_2,q_1,q_2\leq \infty$, $s=|\alpha|\in \Zbb_{\geq 1}$, and
\eqn{calcpropB.1}{
\frac{1}{p_1}+\frac{1}{p_2} = \frac{1}{q_1} + \frac{1}{q_2} = \frac{1}{r}.
}
Then\footnote{Here, we are using the standard notation $\norm{D^s u}^p_{L^p} = \sum_{|\alpha|=s} \norm{D^\alpha u}_{L^p}^p$.}
\alin{calcpropB.2}{
\norm{D^\alpha(uv)}_{L^r} \lesssim \norm{D^s u}_{L^{p_1}}\norm{v}_{L^{q_1}} + \norm{u}_{L^{p_2}}\norm{D^s v}_{L^{q_2}} \label{clacpropB.2.1}
\intertext{and}
\norm{D^\alpha(uv)-uD^\alpha v}_{L^r} \lesssim \norm{Du}_{L^{p_1}}\norm{v}_{W^{s-1,q_1}} + \norm{Du}_{
W^{s-1,p_2}}\norm{v}_{L^{q_2}}
}
for all $u,v \in C^\infty(\Tbb^n)$.
\item[(ii)]  If $s_1,s_2,s_3\in \Zbb_{\geq 0}$, $\;s_1,s_2\geq s_3$,  $1\leq p \leq \infty$, and $s_1+s_2-s_3 > n/p$, then
\eqn{calcpropB.3}{
\norm{uv}_{W^{s_3,p}} \lesssim \norm{u}_{W^{s_1,p}}\norm{v}_{W^{s_2,p}}
}
for all $u\in W^{s_1,p}(\Tbb^n)$ and $v\in W^{s_2,p}(\Tbb^n)$.
\end{enumerate}
\end{thm}

\begin{thm}{\emph{[Moser's estimates]}}  \label{calcpropC}
Suppose $s\in \Zbb_{\geq 1}$, $1\leq p \leq \infty$, $|\alpha|\leq s$, $f\in C^s(\Rbb)$, $f(0) = 0$,
and $V$ is open and bounded in $\Rbb$. Then
\eqn{calcpropC.1}{
\norm{D^\alpha f(u)}_{L^{p}} \leq C\bigl(\norm{f}_{C^s(\overline{V})}\bigr)(1+\norm{u}^{s-1}_{L^\infty})\norm{u}_{W^{s,p}}
}
for all $u \in C^0(\Tbb^n)\cap L^\infty(\Tbb^n)\cap W^{s,p}(\Tbb^n)$ with
$u(\xv) \in V$ for all $\xv\in \Tbb^n$.
\end{thm}

\begin{lem} {\emph{[Ehrling's lemma]}} \label{Ehrling}
Suppose $1\leq p < \infty$, $s_0,s,s_1\in \Zbb_{\geq 0}$, and $s_0 < s < s_1$. Then for any $\epsilon>0$ there exists a constant $C=C(\epsilon)$ such
that
\eqn{Ehrling1}{
\norm{u}_{W^{s,p}} \leq \epsilon \norm{u}_{W^{s_1,p}} + C(\epsilon)\norm{u}_{W^{s_0,p}}
}
for all $u\in W^{s_1,p}(\Tbb^n)$.
\end{lem}

\sect{Hyp}{Symmetric hyperbolic systems}

In this appendix, we analyze the initial value problem for symmetric hyperbolic systems of the form:
\lalin{symivp}{
B^i\del{i}u & = \frac{1}{t}\Bc\Pbb u + H + G\hspace{0.4cm} \text{in $[T_0,T_1)\times \Tbb^n$} \label{symivp.1}, \\
u &= u_0 \hspace{2.4cm} \text{in $\{T_0\}\times \Tbb^n$}, \label{symivp.2}
}
where
\begin{enumerate}[(i)]
\item $T_0<T_1\leq 0$,
\item $\Pbb$ is a constant, symmetric projection operator, i.e. $\Pbb^2=\Pbb$, $\Pbb^T=\Pbb$ and $\del{i}\Pbb=0$,
\item  $v=v(t,x)$ is a $\Rbb^M$-valued map, $u=u(t,x)$, $G(t,x)$ and $H=H(t,x,v,u)$ are $\Rbb^N$-valued maps, $H \in C^0\bigl([T_0,0], C^\infty(\Tbb^n \times \Rbb^M \times \Rbb^N)\bigr)$
and satisfies $H(t,x,v,0)=0$,
\item $B^i=B^i(t,x,v,u)$, and $\Bc=\Bc(t,x,v,u)$ are $\Mbb{N}$-valued maps, and $B^I,\Bc\in  C^0\bigl([T_0,0], C^\infty(\Tbb^n \times \Rbb^M\times \Rbb^N)\bigr)$,  $B^0\in  C^1\bigl([T_0,0], C^\infty(\Tbb^n \times \Rbb^M\times \Rbb^N)\bigr)$
and they satisfy
\eqn{symmetric}{
(B^i)^{T}=B^i \AND [\Pbb,\Bc]=0,
}
\item there exists constants $\kappa,\gamma_1,\gamma_2>0$ such that
\leqn{A0bnd}{
\frac{1}{\gamma_1} \id \leq  B^0(t,x,v,u) \leq \frac{1}{\kappa} \Bc(t,x,v,u) \leq \gamma_2 \id
}
for all $(t,x,v,u) \in [T_0,0]\times \Tbb^n\times\Rbb^M\times \Rbb^N$,
\item
\leqn{Bocross}{
\Pbb^\perp B^0(t,x,v,\Pbb^\perp u)
 \Pbb = \Pbb B^0(t,x,v,\Pbb^\perp u)\Pbb^\perp = 0
}
for all $(t,x,v,u) \in [T_0,0]\times \Tbb^n\times \Rbb^M\times \Rbb^N$ where
\eqn{Pbbperp}{
\Pbb^\perp = \id -\Pbb
}
is the complementary projection operator,\footnote{In other words, $\Pbb^\perp B^0(t,x,v,u)\Pbb = \Pbb^\perp[\Bt^0(t,x,v,u)\cdot\Pbb u]\Pbb$
and $\Pbb B^0(t,x,v,u)\Pbb^\perp = \Pbb[\Bh^{0}(t,x,v,u)\cdot\Pbb u]\Pbb$ for matrix-valued maps
$\Bt^0(t,x,v,u)$ and $\Bh^0(t,x,v,u)$ with the same regularity as $B^0(t,x,v,u)$.}
\item and there exists constants $\theta,\beta_1,\beta_2,\beta_3,\beta_4 \geq 0$ and $\omega > 0$ such that
\lalin{Bodot}{
\bigl|\Pbb^\perp \bigl[D_u B^0(t,x,v,u)\cdot (B^0(t,x,v,u))^{-1}\Bc(t,x,v,u)\Pbb u\bigr]\Pbb^\perp\bigr|_{\text{op}} &\leq |t|\theta +  \frac{2\beta_1}{\omega+|\Pbb^\perp u|^2} |\Pbb u|^2, \label{Bodot.1} \\
\bigl|\Pbb^\perp \bigl[D_u B^0(t,x,v,u)\cdot (B^0(t,x,v,u))^{-1}\Bc(t,x,v,u)\Pbb u\bigr]\Pbb\bigr|_{\text{op}} &\leq |t|\theta +  \frac{2\beta_2}{\sqrt{\omega+|\Pbb^\perp u|^2}} |\Pbb u|, \label{Bodot.2}\\
\bigl|\Pbb \bigl[D_u B^0(t,x,v,u)\cdot (B^0(t,x,v,u))^{-1}\Bc(t,x,v,u)\Pbb u\bigr]\Pbb^\perp\bigr|_{\text{op}} &\leq |t|\theta +  \frac{2\beta_3}{\sqrt{\omega+|\Pbb^\perp u|^2}} |\Pbb u|, \label{Bodot.3}
\intertext{and}
\bigl|\Pbb \bigl[D_u B^0(t,x,v,u)\cdot (B^0(t,x,v,u))^{-1}\Bc(t,x,v,u)\Pbb u\bigr]\Pbb\bigr|_{\text{op}} &\leq |t|\theta +  \beta_4 \label{Bodot.4}
}
for all $(t,x,v,u) \in [T_0,0]\times \Tbb^n\times\Rbb^M\times \Rbb^N$.
\end{enumerate}

Before proceeding, we make some observations. First, we note that the bound \eqref{A0bnd} implies that  $B^0$ and $\Bc$ are invertible and satisfy the matrix operator
bounds
\eqn{Bcop}{
|(B^0(t,x,v,u))^{-1}|_{\text{op}} \leq \gamma_1 \AND |\Bc(t,x,v,u)^{-1}|_{\text{op}} \leq \frac{\gamma_1}{\kappa}
}
for all $(t,x,v,u) \in [T_0,0]\times \Tbb^n\times\Rbb^M\times \Rbb^N$.
Second, a straightforward computation shows that the condition \eqref{Bocross}
is equivalent to
\leqn{Bocross1}{
\Pbb B^0(t,x,v,\Pbb^\perp u)^{-1}\Pbb^\perp =
\Pbb^\perp B^0(t,x,v,\Pbb^\perp u)^{-1} \Pbb   = 0
}
for all $(t,x,u) \in [T_0,0]\times \Tbb^n\times \Rbb^N$. We also note that it follows immediately
from the bounds \eqref{A0bnd} that
there exists
constants $\kappat, \gammat_1, \gammat_2 >0$, where $\gammat_1\leq \gamma_1$, $\kappa \leq \kappat$ and $\gammat_2 \leq \gamma_2$, such that
\leqn{A1bnd}{
\frac{1}{\gammat_1} \Pbb \leq  \Pbb B^0(t,x,v,u)\Pbb \leq \frac{1}{\kappat} \Pbb\Bc(t,x,v,u)\Pbb \leq \gammat_2 \Pbb
}
for all $(t,x,v,u) \in [T_0,0]\times \Tbb^n\times\Rbb^M\times \Rbb^N$.

\begin{thm} \label{symthm}
Suppose $k \in \Zbb_{>n/2+1}$, $u_0\in H^k(\Tbb^n)$, $\;v,G\in C^0\bigl([T_0,0],H^k(\Tbb^n)\bigr)$, assumptions (i)-(vii) are fulfilled, and the constants $\{\beta,\kappat,\gammat_1\}$
satisfy $0\leq \beta\gammat_1 < \kappat$ where $\beta=\sum_{i=1}^4 \beta_4$. Then there exists a $T_*\in (T_0,0)$,
and a unique classical solution $u\in C^1([T_0,T_*)\times \Tbb^n)$ that satisfies
$u\in C^0([T_0,T_*),H^k)\cap C^1([T_0,T_*),H^{k-1})$, the energy estimate
\eqn{symthm1}{
\norm{u(t)}_{H^k}^2+ \norm{G}_{L^\infty([T_0,0),H^k)}^2  - \int_{T_0}^t \frac{1}{\tau} \norm{\Pbb u(\tau)}_{H^k}^2 \,   \leq Ce^{C(t-T_0)}\Bigl(\norm{u(T_0)}_{H^k}^2+ \norm{G}_{L^\infty([T_0,0),H^k)}^2 \Bigr)
}
for $T_0 \leq t < T_*$ where
\eqn{symthm1a}{
C=C\bigl(\norm{u}_{L^\infty([T_0,T_*),H^k)},\norm{v}_{L^{\infty}([T_0,0),H^k)},\norm{G}_{L^\infty([T_0,0),H^k)}^2,\norm{\del{t}v}_{L^{\infty}([T_0,0),H^{k-1})},\theta,\gamma_1,\gamma_2,\kappat,\beta,\omega\bigr),
}
and can be uniquely continued to a larger time interval $[T_0,T^*)$ for some $T^* \in (T_*, 0)$ provided that  $\norm{u}_{L^\infty([T_0,T_*),W^{1,\infty})}  < \infty$.

\smallskip

Moreover, for any $R>0$,
there exists a
\eqn{symth1b}{
\delta= \delta\bigl(R,\norm{v}_{L^\infty([T_0,0),H^k)},\norm{\del{t}v}_{L^\infty([T_0,0),H^{k-1})},\theta,\gamma_1,\gamma_2,\beta,\kappat,\omega\bigr)>0
}
such that if
$\norm{u_0}_{H^k}+\norm{G}_{L^\infty([T_0,0),H^k)} \leq \delta$,
 then the solution $u(t,x)$ exists on the time interval
$[T_0,0)$ and can be uniquely extended to $[T_0,0]$ as an element of $C^0([T_0,0],H^{k-1})$
satisfying
\eqn{symthm1ca}{
\norm{u}_{L^\infty([T_0,0],W^{1,\infty})} \leq R,
}
and
\alin{symthm1c}{
\norm{\Pbb u(t)}_{H^{k-1}} &\lesssim \begin{cases}
-t &  \text{if $\kappat-\beta_4\gammat_1 > 1$} \\
(-t)^{\kappat-\beta_4\gammat_1-\sigma} & \text{if $0<\sigma < \kappat-\beta_4\gammat_1 \leq 1$} \end{cases}, \\
\norm{\Pbb^\perp u(t) - \Pbb^\perp u(0)}_{H^{k-1}} &\lesssim \begin{cases}  -t
& \text{if $\kappat-\beta_4\gammat_1 > 1$ or $[B^0,\Pbb]=0$} \\
 -t +  (-t)^{2(\kappat-\beta_4\gammat_1-\sigma)} &  \text{if $0<\sigma < \kappat-\beta_4\gammat_1 \leq 1$ }
 \end{cases},
}
for $T_0\leq t \leq 0$.

\end{thm}
\begin{proof}
First, the existence, uniqueness and continuation statements follow from standard results; for example, see
\cite[Ch.16 \S 1]{TaylorIII:1996}. Therefore given
a solution $u=u(t,x)$ defined on the time interval $[T_0,T_*)$,  $T_*\in (T_0,1)$, to \eqref{symivp.1}-\eqref{symivp.2}, we act on
\eqref{symivp.1} on the left by $D^\alpha \Bc^{-1}$ to obtain
\eqn{symthm2}{
B^0\del{t} D^\alpha u + B^I\del{I} D^\alpha u = \frac{1}{t}\Bc D^\alpha \Pbb u - \Bc [D^\alpha,\Bc^{-1} B^0]\del{t}u
-\Bc[D^\alpha,\Bc^{-1}B^I]\del{I}u + \Bc D^\alpha(\Bc^{-1}(H+G)).
}
Using \eqref{symivp.1}, we can write this as
\lalin{symthm3}{
B^0\del{t} D^\alpha u + B^I\del{I}& D^\alpha u = \frac{1}{t}\Bigl[\Bc D^\alpha \Pbb u  - \Bc[D^\alpha,\Bc^{-1}B^0](B^0)^{-1}\Bc \Pbb u \Bigr]
+ \Bc[D^\alpha,\Bc^{-1}B^0](B^0)^{-1}B^I\del{I}u  \notag \\
& - \Bc[D^\alpha,\Bc^{-1}B^0](B^0)^{-1}(H+G)
-\Bc[D^\alpha,\Bc^{-1}B^I]\del{I}u + \Bc D^\alpha(\Bc^{-1}(H+G)). \label{symthm3.1}
}
As far as energy estimates are concerned, the only potential problematic term is the one with the coefficient $\frac{1}{t}$ that becomes singular in
the limit $t\nearrow 0$. Since the rest of the terms can be
estimated using well known techniques, again see \cite[Ch.16 \S 1]{TaylorIII:1996}, we will only estimate this potentially singular term in any detail.

Setting $\alpha=0$ in \eqref{symthm3.1}, we obtain, in the usual fashion for symmetric hyperbolic
systems, the energy estimate
\leqn{symthm4}{
\frac{1}{2} \del{t} \ip{u}{B^0 u} = \frac{1}{t}\ip{u}{\Bc \Pbb u} + \frac{1}{2}\ip{u}{\del{0} B u}
 +\frac{1}{2} \ip{u}{\Div B u} + \ip{u}{H}+\ip{u}{G}
}
where
\eqn{symthm5}{
\Div B = \del{I}B^I.
}
Defining the energy norm
\eqn{symthm6}{
\nnorm{u}^2_s =\sum_{|\alpha|\leq s} \ip{D^\alpha u}{B^0 D^\alpha u},
}
we obtain from \eqref{symthm4} and the bounds \eqref{A0bnd} and \eqref{A1bnd} the estimate
\leqn{symthm8}{
 \del{t} \nnorm{u}_0^2 \leq \frac{2\kappat}{t}\nnorm{\Pbb u}_0^2  + \ip{u}{\del{t}B^0 u}+ \gamma_1\norm{\Div B}_{L^\infty} \nnorm{u}_0^2 + 2\sqrt{\gamma_1}\big(\norm{H}_{L^2}+\norm{G}_{L^2}\bigr) \nnorm{u}_0, \quad
 T_0 \leq t < T_*.
}

Using \eqref{symivp.1}, we can write
\leqn{symthm8a}{
 \ip{u}{\del{t}B^0 u} = \ip{u}{\widetilde{\del{t}B^0} u} + \frac{1}{t}\bigl\langle u\bigr|\bigl[D_u B^0(t,\cdot,v,u)\cdot (B^0(t,\cdot,v,u))^{-1}\Bc(t,\cdot,v,u)\Pbb u\bigr]u\rangle
}
where
\alin{symthm8b}{
\widetilde{\del{t}B^0} = (\del{t}&B^0)(t,x,v,u)+D_v B^0(t,x,v,u)\cdot \del{t}v \\
& + D_u B^0(t,x,v,u)\cdot\Bigl[(B^0(t,x,v,u))^{-1}\bigr(-B^I(t,x,v,u)\del{I}u+H(t,x,v,u)+G(t,x)\bigl)\Bigr].
}
Observing that
\alin{symthm8ba}{
\bigl\langle u\bigr|\bigl[D_u B^0\cdot (B^0)^{-1}\Bc\Pbb u\bigr]u\rangle &= \bigl\langle \Pbb^\perp u\bigr|\Pbb^\perp\bigl[D_u B^0\cdot (B^0)^{-1}\Bc\Pbb u\bigr]\Pbb^\perp\Pbb^\perp u\rangle +
\bigl\langle \Pbb^\perp u\bigr|\Pbb^\perp \bigl[D_u B^0\cdot (B^0)^{-1}\Bc\Pbb u\bigr]\Pbb \Pbb u\rangle \\
& + \bigl\langle \Pbb u\bigr|\Pbb \bigl[D_u B^0\cdot (B^0)^{-1}\Bc\Pbb u\bigr]\Pbb^\perp \Pbb^\perp u\rangle + \bigl\langle \Pbb u\bigr|\Pbb\bigl[D_u B^0\cdot (B^0)^{-1}\Bc\Pbb u\bigr]\Pbb \Pbb u\rangle,
}
we see from \eqref{Bodot.1}-\eqref{Bodot.4} that
\leqn{symthm8bb}{
|\bigl\langle u\bigr|\bigl[D_u B^0\cdot (B^0)^{-1}\Bc\Pbb u\bigr]u\rangle| \leq 4|t|\theta\norm{u}^2_{L^2} + \beta\norm{\Pbb u}^2_{L^2},
}
where $\beta = \sum_{i=1}^4 \beta_i$.
From \eqref{symthm8}, \eqref{symthm8a} and \eqref{symthm8bb}, we see with the help of \eqref{A0bnd} that
\leqn{symthm8c}{
 \del{t} \nnorm{u}_0^2 \leq \frac{2(\kappat-\beta\gammat_1)}{t}\nnorm{\Pbb u}_0^2  + \gamma_1\Bigl(4\theta+\norm{\widetilde{\del{t}B^0}}_{L^\infty}+\norm{\Div B}_{L^\infty}\Bigr) \nnorm{u}_0^2 + 2\sqrt{\gamma_1}\big(\norm{H}_{L^2}+\norm{G}_{L^2}\bigr) \nnorm{u}_0
}
for $T_0 \leq t < T_*$.

Next, using $[\Pbb,\Bc]=0$, we observe, for $|\alpha|\leq k$, that
\lalin{symthm9}{
&\ip{D^\alpha u}{\Bc[D^\alpha,\Bc^{-1}B^0](B^0)^{-1}\Bc \Pbb u}  = \ip{D^\alpha \Pbb^{\perp}u }{\Bc[D^\alpha,\Bc^{-1}\Pbb^\perp B^0 \Pbb](B^0)^{-1}\Bc \Pbb u}  \notag\\
& \hspace{2.0cm}  + \ip{D^\alpha \Pbb^\perp  u }{\Bc[D^\alpha,\Bc^{-1}\Pbb^\perp  B^0]\Pbb^\perp (B^0)^{-1}\Pbb\Bc \Pbb u}  + \ip{D^\alpha \Pbb u}{\Bc[D^\alpha,\Bc^{-1}B^0](B^0)^{-1}\Bc \Pbb u}.
\notag
}
Applying the Sobolev, commutator, product and Moser calculus inequalities, i.e. Theorems \ref{Sobolev}, \ref{calcpropB}
and \ref{calcpropC}, to the above expression, we find
with the help of \eqref{Bocross} and \eqref{Bocross1} the estimate
\eqn{symthm10a}{
\ip{D^\alpha u}{\Bc[D^\alpha,\Bc^{-1}B^0](B^0)^{-1}\Bc \Pbb u} \leq C\bigl(K_1,K_2\bigr)\norm{\Pbb u}_{H^{k}}\norm{\Pbb u}_{H^{k-1}}
}
where
\eqn{K1def}{
K_1 = \norm{u}_{L^\infty([T_0,T_*),H^k)} \AND  K_2 =\norm{v}_{L^\infty([T_0,0),H^k)}.
}
Hence, for any $\epsilon > 0$, we find that
\leqn{symthm10}{
\ip{D^\alpha u}{\Bc[D^\alpha,\Bc^{-1}B^0](B^0)^{-1}\Bc \Pbb u} \leq C\bigl(K_1,K_2\bigr)\Bigl(\epsilon \norm{\Pbb u}_{H^{k}}^2 + c(\epsilon)\norm{\Pbb u}_{L^2}^2 \Bigr)
}
by Ehrling's lemma, Lemma \ref{Ehrling}, and Young's inequality in the form $ab \leq \frac{1}{2r}a^2 + \frac{r}{2}b^2$
for $a,b\geq 0$ and $r>0$. We also see from
\alin{symthm10a}{
&\bigl\langle D^\alpha u\bigr|\bigl[D_u B^0\cdot (B^0)^{-1}\Bc\Pbb u\bigr]D^\alpha u\rangle = \\
& \hspace{1.5cm} \bigl\langle \Pbb^\perp D^\alpha u\bigr|\Pbb^\perp\bigl[D_u B^0\cdot (B^0)^{-1}\Bc\Pbb u\bigr]\Pbb^\perp\Pbb^\perp D^\alpha u\rangle
 +
\bigl\langle \Pbb^\perp D^\alpha u\bigr|\Pbb^\perp \bigl[D_u B^0\cdot (B^0)^{-1}\Bc\Pbb u\bigr]\Pbb \Pbb D^\alpha u\rangle \\
& \hspace{2.5cm} + \bigl\langle \Pbb D^\alpha u\bigr|\Pbb \bigl[D_u B^0\cdot (B^0)^{-1}\Bc\Pbb u\bigr]\Pbb^\perp \Pbb^\perp D^\alpha u\rangle
  + \bigl\langle \Pbb D^\alpha u\bigr|\Pbb\bigl[D_u B^0\cdot (B^0)^{-1}\Bc\Pbb u\bigr]\Pbb \Pbb D^\alpha u\rangle,
}
the bounds \eqref{Bodot.1}-\eqref{Bodot.4}, and the calculus inequalities that
\alin{symthm10b}{
\bigl|\bigl\langle D^\alpha u\bigr|\bigl[&D_u B^0\cdot (B^0)^{-1}\Bc\Pbb u\bigr]D^\alpha u\rangle\bigr| \leq 4|t|\theta \norm{D^\alpha u}_{L^2}^2 \\
&+C(K_1,\omega)\bigl(\beta_1\norm{\Pbb u}_{L^2}^2 + (\beta_2+\beta_3)\norm{\Pbb u}_{L^2}\norm{\Pbb D^\alpha u}_{L^2}\bigr) + 2\beta_4 \norm{\Pbb D^\alpha u}_{L^2}^2. 
}
Applying Young's inequality to the this expression, we see for any $\epsilon >0$ that
\leqn{symthm10c}{
\bigl|\bigl\langle D^\alpha u\bigr|\bigl[D_u B^0\cdot (B^0)^{-1}\Bc\Pbb u\bigr]D^\alpha u\rangle\bigr| \leq 4|t|\theta \norm{u}_{H^k}^2
+ C(K_1,\theta,\beta,\omega)\bigl(\epsilon \norm{\Pbb u}_{H^k}^2 + c(\epsilon)\norm{\Pbb u}_{L^2}^2\bigr) + 2\beta \norm{\Pbb D^\alpha u}_{L^2}^2.
}
Setting
\eqn{nudef}{
\nu = \norm{G}_{L^\infty([T_0,0),H^k)}
}
and
applying the standard energy estimates, i.e. as in the derivation of \eqref{symthm4}, for symmetric hyperbolic systems to \eqref{symthm3.1}, we obtain, after summing over $|\alpha|\leq k$ and using \eqref{symthm10}
 and \eqref{symthm10c} together with the
calculus inequalities,
the estimate\footnote{Here, we are taking essentially the same approach to deriving energy
estimates for symmetric hyperbolic systems as in \cite[Ch.16 \S 1]{TaylorIII:1996}.}
\leqn{symthm11}{
 \frac{1}{2}\del{t} \nnorm{u}_k^2 \leq \frac{\kappat-\beta\gammat_1}{t} \nnorm{\Pbb u}_k^2 +
C(K_1,K_2,K_3,\theta,\gamma_1,\beta,\omega,\nu)\biggl[ -\frac{1}{t}\Bigl(\epsilon \nnorm{\Pbb u}_{k}^2 + c(\epsilon)\nnorm{\Pbb u}_{0}^2 \Bigr) + \nnorm{u}_{k}^2 + \nu\biggr]
}
for $T_0 \leq t < T_*$, where
\eqn{symthm11a}{
K_3 = \norm{\del{t}v}_{L^\infty([T_0,0),H^{k-1})}.
}

Since $\kappat-\beta/\gammat_1 > 0$ by assumption, we see, after choosing $\epsilon$ small enough and adding a suitable multiple of \eqref{symthm8c} to \eqref{symthm11}, that
\eqn{symthm13a}{
 \del{t}\biggl( \nnorm{u}_k^2 + c \nnorm{u}^2_0 +\nu^2 - \int_{T_0}^t \frac{1}{\tau} \nnorm{\Pbb u(\tau)}_k^2 \, d\tau \biggr) \leq C(K_1,K_2,K_3,\theta,\gamma_1,\beta,\omega,\nu)\bigl(\nnorm{u}_{k}^2+\nu^2)
}
where $c=c(K_1,K_2,K_3,\theta,\gamma_1,\beta,\kappat,\omega,\nu)$.
Applying Gronwall's inequality yields
\eqn{symthm13b}{
\norm{u(t)}_{H^k}^2+\nu^2 - \int_{T_0}^t \frac{1}{\tau} \norm{\Pbb u(\tau)}_{H^k}^2 \,   \leq Ce^{C(t-T_0)}\bigl(\norm{u(T_0)}_{H^k}^2+\nu^2), \qquad T_0 \leq t < T_*,
}
for some  constant $C(K_1,K_2,K_3,\theta,\gamma_1,\gamma_2,\beta,\kappat,\omega,\nu)$. From this, Sobolev's inequality and the continuation principle, it
follows, for any fixed $R>0$, that there exists
a
\eqn{symthm13ba}{
\delta=\delta\bigl(R,K_2,K_3,\theta,\gamma_1,\gamma_2,\beta,\kappat,\omega\bigr) > 0
}
such that the solution $u=u(t,x)$ exists on the time interval $[T_0,0)$ and satisfies
\leqn{symthm13c}{
\norm{u}_{L^\infty([T_0,0),H^k)}+\nu + \biggl(- \int_{T_0}^0 \frac{1}{\tau} \norm{\Pbb u(\tau)}_{H^k}^2\, d\tau\biggr)^{\frac{1}{2}} \leq R
}
provided that $\norm{u(T_0)}_{H^k}+\nu < \delta$.

To complete the proof, we assume that the condition $\norm{u(T_0)}_{H^k}+\nu < \delta$ holds.
Writing \eqref{symivp.1} as
\eqn{symthm14a}{
\del{t}u = (B^0)^{-1}\biggl(B^I\del{I}u
+ \frac{1}{t}\Bc \Pbb u + H+ G \biggr),
}
and multiplying on the left by $\Pbb^\perp$, we obtain
\leqn{symthm14}{
\del{t}\Pbb^\perp u = \Pbb^\perp(B^0)^{-1} \Bigl(B^I\del{I}u
 + H+ G \Bigr) + \frac{1}{t}\Pbb^\perp (B^0)^{-1}\Pbb \Bc \Pbb u.
}
Integrating this in time, we see after applying the $H^{k-1}$ norm that
\lalin{symthm15}{
\norm{\Pbb^\perp u(t_2) - \Pbb^\perp u(t_1)}_{H^{k-1}} \leq  \int_{t_1}^{t_2}& \norm{\Pbb^\perp(B^0(\tau))^{-1}B^I(\tau)\del{I}u(\tau)}_{H^{k-1}} \notag \\
+ \norm{\Pbb^\perp(B^0(\tau))^{-1}H(\tau)}_{H^{k-1}}& +  \norm{\Pbb^\perp(B^0(\tau))^{-1}G(\tau)}_{H^{k-1}}
- \frac{1}{\tau}\norm{\Pbb^\perp (B^0(\tau))^{-1}\Pbb\Bc \Pbb u(\tau)}_{H^{k-1}} \, d\tau \label{symthm15.1}
}
for any $t_1,t_2$ satisfying $T_0 \leq t_1 < t_2$. Using the calculus inequalities, i.e. product, Sobolev, H\"{o}lder,
and Moser, in conjunction with \eqref{Bocross1} and the energy estimates \eqref{symthm13c}, we obtain from the
above inequality that
\eqn{symthm16}{
\norm{\Pbb^\perp u(t_2) - \Pbb^\perp u(t_1)}_{H^{k-1}} \leq C(R,K_2)\biggl(  |t_2-t_1|
- \int_{t_t}^{t_2}\frac{1}{\tau} \norm{\Pbb u(\tau)}_{H^{k-1}}^2 \, d\tau\biggr) \leq C(R,K_2)\bigl( |t_2-t_1| + \text{o}\bigl(|t_2-t_1|\bigr)\bigr).
}
It follows directly from this inequality that $\lim_{t\nearrow 0}\Pbb^\perp u(t)$ exists in $H^{k-1}(\Tbb^n)$, and moreover, that
\eqn{symthm16a}{
\Pbb^\perp u \in C^0([T_0,0],H^{k-1}).
}

Next, we apply the $H^{k-1}$ norm to \eqref{symthm14} and then integrate in time. With the help of the calculus
inequalities, \eqref{Bocross1}, and \eqref{symthm13c}, this yields the estimate
\leqn{symthm17}{
\int_{T_0}^t\norm{\del{t}\Pbb^\perp u(\tau)}_{H^{k-1}}\, d\tau \leq C(R,K_2) \qquad T_0\leq t < 0.
}
Multiplying \eqref{symivp.1} on the left by $\Pbb$ shows that
$\Pbb u$
satisfies
\leqn{symthm19}{
\Bb^i\del{i}\Pbb u = \frac{1}{t}\bar{\Bc}\Pbb u + \Hb + \Gb
}
where
\gath{symthm20}{
\Bb^i = \Pbb B^i \Pbb, \quad \bar{\Bc}  = \Pbb \Bc \Pbb, \quad
\Gb  = \Pbb G
\intertext{and}
\Hb = \Pbb H + \Pbb B^i\Pbb^\perp \del{i}\Pbb^\perp u.
}

Applying the same arguments used to derive the energy estimate \eqref{symthm8} to the equation \eqref{symthm19} yields
\leqn{symthm21}{
 \frac{1}{2}\del{t} \nnorm{\Pbb u}_0^2 \leq \frac{\kappat-\beta_4\gammat_1}{t}\nnorm{\Pbb u}_0^2  + \frac{\gammat_1}{2}\bigl(\theta+\norm{\widetilde{\del{t}\Bb^0}}_{L^\infty}+\norm{\Div \Bb}_{L^\infty}\bigr) \nnorm{\Pbb u}_0^2 + \sqrt{\gammat_1}\big(\norm{\Hb}_{L^2}+\norm{\Gb}_{L^2}\bigr) \nnorm{\Pbb u}_0
}
for $T_0 \leq t < 0$, where now the energy norm $\nnorm{\cdot}$ is defined by
\eqn{symthm22}{
\nnorm{u}_s^2 = \sum_{|\alpha|\leq s} \ip{D^\alpha u}{\Bb^0 D^\alpha u},
}
and
\alin{symthm22a}{
\widetilde{\del{t}\Bb^0} = (\del{t}&\Bb^0)(t,x,v,u)+D_v \Bb^0(t,x,v,u)\cdot \del{t}v \\
& + D_u \Bb^0(t,x,v,u)\cdot\Bigl[(B^0(t,x,v,u))^{-1}\bigr(-B^I(t,x,v,u)\del{I}u+H(t,x,v,u)+G(t,x)\bigl)\Bigr].
}
Using  \eqref{Bocross} and  \eqref{symthm13c}, we see that \eqref{symthm21} implies that
\eqn{symthm23}{
\del{t} \nnorm{\Pbb u}_0 \leq \frac{\kappat-\beta_4\gammat_1}{t}\nnorm{\Pbb u}_0  + C(R,K_2,K_3,\theta,\gamma_1,\nu)\Bigl[\bigl(1+\norm{\del{t}\Pbb^\perp u}_{H^{k-1}}\bigr) \nnorm{\Pbb u}_0 + \norm{u}_{H^1}+\norm{\Gb}_{L^2}\Bigr]
}
for $T_0 \leq t < 0$.
It then follows from Gronwall's inequality\footnote{Here, we are using the following form of Gronwall's inequality: if $x(t)$ satisfies
$x'(t) \leq a(t)x(t) +h(t)$, $t\geq T_0$, then $x(t)\leq x(T_0)e^{A(t)}+ \int_{T_0}^t e^{A(t)-A(\tau)}h(\tau) \, d\tau$ where
$A(t)= \int_{T_0}^t a(\tau)\, d\tau$. In particular, we observe from this that if $a(t)=\frac{\lambda}{t}+b(t)$, $\lambda \in \Rbb$, where
$|b(t)|\leq r$ then
\eqn{gronwall}{
x(t) \leq e^{r(t-T_0)}x(T_0)\left(\frac{t}{T_0}\right)^\lambda +   e^{r(t-2T_0)}(-t)^\lambda \int_{T_0}^t \frac{|h(\tau)|}{(-\tau)^\lambda}\, d\tau
}
for $T_0 \leq t < 0$.
} and the bounds \eqref{symthm13c} and \eqref{symthm17} that
\leqn{symthm24}{
\norm{\Pbb u(t)}_{L^2} \lesssim \begin{cases}
-t & \text{if $\kappat-\beta_4\gammat_1 > 1$} \\
t \ln\bigl(\frac{t}{T_0}\bigr) & \text{if $\kappat-\beta_4\gammat_1 = 1$ } \\
(-t)^{\kappat-\beta_4\gammat_1} & \text{if $\kappat-\beta_4\gammat_1 < 1$} \end{cases}, \quad T_0 \leq t < 0.
}

Proceeding as above, we can apply the same arguments used to derive the energy estimate \eqref{symthm11} to \eqref{symthm19} to obtain the estimate
\eqn{symthm25}{
 \frac{1}{2}\del{t} \nnorm{\Pbb u}_{k-1}^2 \leq \frac{\kappat-\beta_4\gammat_1}{t} \nnorm{\Pbb u}_{k-1}^2 +
C(R,K_2,K_3,\theta,\gamma_1,\nu)\biggl[ -\frac{1}{t}\Bigl(\epsilon \nnorm{\Pbb u}_{k-1}^2 + c(\epsilon)\nnorm{\Pbb u}_{0}^2 \Bigr) + \nnorm{\Pbb u}_{k-1}\biggr].
}
Choosing $\epsilon$ small enough, we see that the inequality
\leqn{symthm25a}{
\del{t} \nnorm{\Pbb u}_{k-1} \leq \frac{\kappat-\beta_4\gammat_1-\sigma}{t} \nnorm{\Pbb u}_{k-1} +
C(R,K_2,K_3,\theta,\gamma_1,\gamma_2,\nu,\beta_4,\kappat)\biggl(\frac{1}{t}\norm{\Pbb u}_{L^2} + 1\biggr)
}
holds for any fixed $\sigma >0$. Choosing
\begin{equation*}
\sigma \in \begin{cases} (0,\kappat-\beta_4\gammat_1-1) & \text{if $\kappat-\beta_4\gammat_1>1$} \\
(0,\kappat-\beta_4\gammat_1) &  \text{if $\kappat-\beta_4\gammat_1\leq 1$}
\end{cases}, 
\end{equation*}
it follows from \eqref{symthm24} and an application of Gronwall's inequality to \eqref{symthm25a}, that $\Pbb u$ satisfies the decay estimate
\eqn{symthm26}{
\norm{\Pbb u(t)}_{H^{k-1}} \lesssim \begin{cases}
-t & \text{if $\kappat -\beta_4\gammat_1> 1$} \\
(-t)^{\kappat-\beta_4\gammat_1-\sigma} & \text{if $\kappat-\beta_4\gammat_1 \leq 1$} \end{cases}, \quad T_0 \leq t < 0.
}
Using the above estimate in conjunction with \eqref{symthm15.1}, we see with the help of the calculus inequalities and the
bound \eqref{symthm13c} that
\lalin{symthm27}{
\norm{\Pbb^\perp u(t_2) - \Pbb^\perp u(t_1)}_{H^{k-1}}& \lesssim
\begin{cases}  |t_2-t_1|
& \text{if $\kappat-\beta_4\gammat_1 > 1$ } \\
  |t_2-t_1| +  (-t_1)^{2(\kappat-\beta_4\gammat_1-\sigma)}-(-t_2)^{2(\kappat-\beta_4\gammat_1-\sigma)}  &  \text{if $\kappat-\beta_4\gammat_1 \leq 1$ }
 \end{cases} \label{symthm27.1}
}
for any $t_1,t_2$ satisfying $T_0\leq t_1<t_2 < 0$. We also note in \eqref{symthm15.1} that if
$[B^0,\Pbb]=0$, then the term with the $\frac{1}{\tau}$ vanishes and we obtain the estimate
\leqn{symthm28}{
\norm{\Pbb^\perp u(t_2) - \Pbb^\perp u(t_1)}_{H^{k-1}} \leq C(R,K_2)|t_2-t_1|, \quad T_0\leq t_1<t_2<0.
}
Sending $t_2\nearrow 0$ in \eqref{symthm27.1} and \eqref{symthm28}, we see that $\Pbb^\perp u$ satisfies
\alin{symthm29}{
\norm{\Pbb^\perp u(t) - \Pbb^\perp u(0)}_{H^{k-1}} \lesssim
& \begin{cases}  -t
& \text{if $\kappat-\beta_4\gamma_1 > 1$ or $[B^0,\Pbb]=0$} \\
 -t +  (-t)^{2(\kappat-\beta_4\gammat_1-\sigma) }  &  \text{if $\kappat-\beta_4\gammat_1 \leq 1$ }
 \end{cases}
}
for $T_0 \leq t < 0$.

\end{proof}

\begin{rem} \label{boundrem}
$\,$

\begin{itemize}
\item[(i)]
The bounds \eqref{A0bnd}-\eqref{Bodot.4}  and \eqref{A1bnd} do not need to hold for all  $(v,u)\in \Rbb^M \times \Rbb^N$. The conclusions of
Theorem \ref{symthm} remain valid provided that
\begin{itemize}
\item[(a)] there exists
constants $\rc,\Rc>0$ such that \eqref{A0bnd}-\eqref{Bodot.4}  and \eqref{A1bnd} hold for all $(t,x,v,u)\in [T_0,0]\times \Tbb^n\times B_{\rc}(\Rbb^M)\times B_{\Rc}(\Rbb^N)$,
\item[(b)]  the map $v\in C^0([T_0,0],H^k(\Tbb^n))$ satisfies $\norm{v}_{L^\infty([T_0,0]\times \Rbb^M)} < \rc$,
\item[(c)] and the continuation principle is modified to having the solution $u(t,x)$ remain in a proper subset of $B_{\Rc}(\Rbb^N)$ for $(t,x)\in[T_0,T_*)\times \Tbb^n$ in addition
to satisfying $\norm{u}_{L^\infty([T_0,T_*),W^{1,\infty})}<\infty$,
\item[(d)]  and the constant $R>0$ is chosen so that $R < \Rc$.
\end{itemize}
\item[(ii)] Given any two constants $\betat_1,\betat_2 >0$, we observe that we can satisfy
\eqn{rembound1}{
C_1 \leq \frac{2\beta_1}{\omega + |\Pbb^\perp u|^2} \AND C_2 \leq \frac{2\beta_2}{\sqrt{\omega + |\Pbb^\perp u|^2}}
}
for all $u\in B_\Rc(\Rbb^N)$
by setting
\eqn{rembound2}{
\beta_1 = \frac{\betat_1(\Rc+\Rc^2)}{2}, \quad \beta_1 = \frac{\betat_2\sqrt{(\Rc+\Rc^2)}}{2}\AND \omega = \Rc.
}
This shows that if the bounds
\lalin{Bodotrem}{
\bigl|\Pbb^\perp \bigl[D_u B^0(t,x,v,u)\cdot (B^0(t,x,v,u))^{-1}\Bc(t,x,v,u)\Pbb u\bigr]\Pbb^\perp\bigr|_{\text{op}} &\leq |t|\theta +  \betat_1 |\Pbb u|^2, \label{Bodotrem.1} \\
\bigl|\Pbb^\perp \bigl[D_u B^0(t,x,v,u)\cdot (B^0(t,x,v,u))^{-1}\Bc(t,x,v,u)\Pbb u\bigr]\Pbb\bigr|_{\text{op}} &\leq |t|\theta +  \betat_2 |\Pbb u|, \label{Bodotrem.2} \\
\intertext{and}
\bigl|\Pbb \bigl[D_u B^0(t,x,v,u)\cdot (B^0(t,x,v,u))^{-1}\Bc(t,x,v,u)\Pbb u\bigr]\Pbb^\perp\bigr|_{\text{op}} &\leq |t|\theta +  \betat_3 |\Pbb u|, \label{Bodotrem.3}
}
hold for all $(t,x,v,u)\in [T_0,0]\times \Tbb^n\times B_{\rc}(\Rbb^M)\times B_{\Rc}(\Rbb^N)$, then the bounds \eqref{Bodot.1}-\eqref{Bodot.2} will be satisfied for all
$(t,x,v,u)\in [T_0,0]\times \Tbb^n\times B_{\rc}(\Rbb^M)\times B_{\Rc}(\Rbb^N)$ for the constants
$\beta_1,\beta_2,\beta_3$ and $\omega$ given by
\eqn{rembound3}{
\beta_1 = \frac{\betat_1(\Rc+\Rc^2)}{2}, \quad \beta_3=\beta_2 = \frac{\betat_2\sqrt{(\Rc+\Rc^2)}}{2}\AND \omega = \Rc.
}
In particular, this shows that by choosing $\Rc$ small enough, we can always arrange that $\beta_1$, $\beta_2$ and $\beta_3$ are as small as we like.

Since
the left-hand side of \eqref{Bodot.2}-\eqref{Bodot.4} are linear in $\Pbb u$, it is clear that there always exists constants $\betat_2,\betat_3,\beta_4>0$ such that the bounds \eqref{Bodotrem.2}-\eqref{Bodotrem.3} and \eqref{Bodot.4}
hold for all $(t,x,v,u)\in [T_0,0]\times \Tbb^n\times B_{\rc}(\Rbb^M)\times B_{\Rc}(\Rbb^N)$. Moreover, by choosing $\Rc>0$ small enough, we can take $\beta_4>0$ as small as we like.

\bigskip

\noindent Thus we can conclude the following:

\medskip
\noindent \textit{If there exists a constant $\betat_1>0$ such that
\eqn{rembound4}{
\bigl|\Pbb^\perp \bigl[D_u B^0(t,x,v,u)\cdot (B^0(t,x,v,u))^{-1}\Bc(t,x,v,u)\Pbb u\bigr]\Pbb^\perp\bigr|_{\text{op}} \leq |t|\theta +  \betat_1 |\Pbb u|^2
}
for all $(t,x,v,u)\in [T_0,0]\times \Tbb^n\times B_{\rc}(\Rbb^M)\times B_{\Rc}(\Rbb^N)$ and $\Rc \in (0,\Rc_0]$, then there exists constants $\beta_1,\beta_2,\beta_3,\beta_4>0$
such that the bounds \eqref{Bodot.1}-\eqref{Bodot.4} hold for all  $(t,x,v,u)\in [T_0,0]\times \Tbb^n\times B_{\rc}(\Rbb^M)\times B_{\Rc}(\Rbb^N)$. Moreover, by choosing $\Rc >0$
small enough, the constants $\beta_1,\beta_2,\beta_3,\beta_4$ can be taken arbitrarily small. }
\end{itemize}
\end{rem}


\bibliographystyle{amsplain}
\bibliography{refs}

\end{document}